\documentclass[12pt]{article}
\usepackage{amsmath,amssymb,amsthm,mathtools}
\usepackage[margin=1in]{geometry}
\usepackage[numbers,sort&compress]{natbib}
\usepackage[colorlinks=true,linkcolor=blue,citecolor=blue,urlcolor=blue]{hyperref}

\newtheorem{theorem}{Theorem}[section]
\newtheorem{lemma}[theorem]{Lemma}
\newtheorem{proposition}[theorem]{Proposition}
\newtheorem{corollary}[theorem]{Corollary}
\theoremstyle{definition}
\newtheorem{definition}[theorem]{Definition}
\newtheorem{assumption}[theorem]{Assumption}
\theoremstyle{remark}
\newtheorem{remark}[theorem]{Remark}
\newtheorem{example}[theorem]{Example}
\newcommand{\cA}{\mathcal{A}}
\newcommand{\cH}{\mathcal{H}}
\newcommand{\cL}{\mathcal{L}}
\newcommand{\Tr}{\operatorname{Tr}}
\newcommand{\supp}{\operatorname{supp}}
\newcommand{\Id}{\mathrm{I}}
\newcommand{\normone}[1]{\|#1\|_1}
\newcommand{\normtwo}[1]{\|#1\|_2}
\newcommand{\diag}{\operatorname{diag}}
\newcommand{\Span}{\operatorname{Span}}
\newcommand{\rank}{rank}
\title{Entropy Moduli and Support-Sensitive BKM Coercivity for Rank-Deficient Non-Commutative Markov Semigroups}
\author{Hassan Nasreddine\thanks{\texttt{hassan.nasreddine@hotmail.com}}}
\date{}

\begin{document}
\maketitle

\begin{abstract}
We study entropy--coherence relations near rank-deficient support boundaries
in finite-dimensional quantum systems.
For block-diagonal reference states, we establish a support-sensitive
coercivity estimate in which the entropy cost of cross-boundary coherence
acquires a logarithmic enhancement determined by the small population scale
near the boundary.

Combined with a finite-time entropy budget, this yields conditional
entropy--activation bounds of order
\[
e^{-\alpha t}(1+\alpha t)^{-1/2}
\]
within a coherence-dominant regime.
The proof proceeds through a pinching reduction and effective \(2\times2\)
BKM estimates adapted to the coherence--population structure.
For Davies semigroups satisfying additional secular and population-rate
assumptions, the resulting bounds apply along trajectories for which the
local support and coherence hypotheses remain valid.
\end{abstract}

\section{Introduction}\label{sec:intro}

Quantum Markov semigroups provide the standard setting for dissipative
dynamics in finite-dimensional quantum systems. In the Davies setting they arise
from weak-coupling limits of open systems and describe thermalization,
decoherence, and entropy production. A basic problem is to convert entropy
control into control of a physically relevant deviation from a reference state.

For faithful reference states this relation is typically governed by quadratic
stability estimates such as Pinsker inequalities or local entropy expansions.
The situation is different for rank-deficient reference states. Such states
occur naturally in the presence of symmetries, superselection rules,
zero-temperature limits, or effective low-energy projections. Writing
\(P=\supp(\sigma)\) and \(Q=I-P\), states near the support boundary may have
small \(Q\)-population
\[
\varepsilon_Q(\rho)=\Tr(Q\rho Q)
\]
together with a nontrivial coherence block \(P\rho Q\). In this regime the
entropy contribution associated with cross-boundary coherence is not captured
well by support-insensitive trace-distance estimates.

A central ingredient is a boundary coercivity estimate for the coherence
entropy associated with the pinching
\[
\Pi\rho=P\rho P+Q\rho Q.
\]
Under a local lower bound on the \(P\)-block and sufficiently small
\(Q\)-population, one obtains a lower bound of the form
\[
D(\rho\|\Pi\rho)
\gtrsim
\|P\rho Q\|_2^2
\log\!\left(\frac{a_0}{\varepsilon_Q(\rho)}\right).
\]
Thus cross-boundary coherence carries a logarithmically enhanced entropy cost
near the support boundary. In coherence-dominant regimes, inversion of the
modulus \(x^2\log(C/x)\) yields entropy-to-activation estimates with a
correction factor of the form
\[
\frac{e^{-\alpha t}}{\sqrt{1+\alpha t}}.
\]
This correction reflects the structure of the entropy modulus and should be
interpreted as a certification consequence rather than as an intrinsic
convergence-rate statement for the underlying dynamics. The same logarithmic
coherence cost also enters relative entropy with respect to faithful
block-diagonal references, including the regularized states
\(\sigma_\varepsilon\) used later.

A key step is a reduction to effective \(2\times2\) BKM blocks via pinching and
data processing. These estimates do not rely on a particular dynamics. The
associated activation functional is
\[
\mathcal A(\rho)
=
\sqrt{\|P\rho Q\|_2^2+\Tr(Q\rho Q)}.
\]
Since block-diagonal states may have nonzero \(Q\)-population while
\(D(\rho\|\Pi\rho)=0\), entropy control alone does not determine
\(\mathcal A(\rho)\). Additional population or coherence information is
therefore required near the support boundary.

The structural estimates are applied to Davies semigroups under additional local
assumptions. Secular decoupling and population estimates provide dynamical
inputs verifying the hypotheses of the entropy bounds on finite time windows.
These results do not constitute general mixing-time or MLSI bounds for Davies
generators.

Sections~\ref{sec:setup}--\ref{sec:BKM} develop the structural entropy
estimates and the associated coercivity bounds.
Section~\ref{sec:entropy-activation} derives entropy-to-activation estimates.
Section~\ref{sec:main} applies these results to Davies semigroups under
additional assumptions. The remaining sections discuss explicit regimes,
comparisons with recoverability bounds, and related consequences.
\subsection{Structure}

Sections~\ref{sec:setup}--\ref{sec:activation} introduce notation and the activation
functional. Sections~\ref{sec:Davies}--\ref{sec:BKM} develop the entropy-budget and
coherence estimates together with the support-sensitive BKM coercivity mechanism.
Section~\ref{sec:entropy-activation} gives the structural no-go theorem,
conditional bound, split bound, and absorbing bound.
Section~\ref{sec:main} applies these results to Davies semigroups under
additional assumptions.
Sections~\ref{sec:window}--\ref{sec:corollaries} give explicit windows and corollaries.
Section~\ref{sec:cps} develops the coherence--population separation theorem and its
comparison with recoverability bounds. Section~\ref{sec:literature} places the work
in context. Section~\ref{sec:discussion} discusses scope and limitations.

\section{Setting and Notation}\label{sec:setup}

Let \(\cH\) be a finite-dimensional Hilbert space with \(d=\dim\cH\).
Let \(\sigma\) be a density matrix of rank \(r<d\), and write
\[
P=\supp(\sigma),
\qquad
Q=\Id-P.
\]
Then \(\sigma|_{P\cH}>0\).

For \(\varepsilon\in(0,1)\), define
\begin{equation}
\sigma_\varepsilon
:=
(1-\varepsilon)\sigma+\frac{\varepsilon}{d}\Id.
\label{eq:sigma-eps}
\end{equation}
The state \(\sigma_\varepsilon\) is faithful and block-diagonal relative to
\(P\oplus Q\), with
\[
Q\sigma_\varepsilon Q=\frac{\varepsilon}{d}Q.
\]

Relative to the decomposition \(P\oplus Q\),
\begin{equation}
\rho
=
\begin{pmatrix}
A & B\\
B^* & C
\end{pmatrix},
\qquad
A=P\rho P,
\quad
B=P\rho Q,
\quad
C=Q\rho Q.
\label{eq:block}
\end{equation}

Davies generators and the associated secular structure are reviewed in
Section~\ref{sec:Davies}.

\section{The Boundary-Activation Functional}\label{sec:activation}

\begin{definition}\label{def:activation}
For any state \(\rho\) with block decomposition~\eqref{eq:block}, define
\begin{align}
c(\rho)
&:=
\normtwo{P\rho Q}^2
=
\normtwo{B}^2,
\label{eq:c-def}\\
\varepsilon_Q(\rho)
&:=
\Tr(Q\rho Q)
=
\Tr(C),
\label{eq:epsQ-def}\\
\cA(\rho)
&:=
\sqrt{c(\rho)+\varepsilon_Q(\rho)}.
\label{eq:A-def}
\end{align}
The coherence-dominance ratio is
\[
R(\rho)^2
:=
\frac{c(\rho)}{\cA(\rho)^2}.
\]
\end{definition}

One has \(\cA(\rho)=0\) if and only if \(B=0\) and \(C=0\), equivalently,
\(\rho\) is supported on \(P\cH\). The Schur complement inequality
\[
BB^*\le \|A\|_\infty C
\]
implies
\[
c(\rho)\le \|A\|_\infty \varepsilon_Q(\rho).
\]
Hence \(c(\rho)>0\) implies \(\varepsilon_Q(\rho)>0\). The associated
logarithmic entropy cost is quantified in Lemma~\ref{lem:BKM}.

\section{Davies Estimates and Coherence Windows}
\label{sec:Davies}

This section records the Davies-semigroup estimates used in the
subsequent entropy--activation analysis.

\begin{assumption}[Secular \(P\)-\(Q\) decoupling]
\label{ass:secular}
Assume that \(P\) and \(Q\) are unions of energy eigenspaces of \(H\).
There exist energy eigenbases \(\{|p\rangle\}_{p\in P}\) of \(P\cH\) and
\(\{|e\rangle\}_{e\in Q}\) of \(Q\cH\) such that
\begin{equation}
\cL(|p\rangle\langle e|)
=
(-\Gamma_{pe}-i\omega_{pe})
|p\rangle\langle e|,
\qquad
\Gamma_{pe}\ge0,
\qquad
\omega_{pe}=E_p-E_e .
\label{eq:secular}
\end{equation}
Set
\[
\Gamma_{\max}:=\max_{p,e}\Gamma_{pe}.
\]
\end{assumption}

Assumption~\ref{ass:secular} is the secular structure used in the Davies
applications below.

\begin{assumption}[Cross-boundary rate parameters]
\label{ass:rates}
The classical transition rates are
\[
W_{nm}
:=
\sum_\alpha
\gamma_\alpha(E_m-E_n)
|\langle n|S_\alpha|m\rangle|^2 ,
\]
where \(W_{nm}\) denotes the rate from \(m\) to \(n\). Define
\begin{equation}
\mu
:=
\max_{p\in P}\sum_{e\in Q}W_{ep},
\qquad
\eta
:=
\min_{e\in Q}\sum_{p\in P}W_{pe},
\label{eq:rates}
\end{equation}
and assume \(\mu+\eta>0\). Set
\[
k:=\mu+\eta,
\qquad
\bar\varepsilon:=\frac{\mu}{\mu+\eta}.
\]
\end{assumption}

\begin{proposition}[Davies estimates]
\label{prop:Davies}
Under Assumptions~\ref{ass:secular}--\ref{ass:rates}, for all \(t\ge0\),
\begin{align}
c(\rho_t)
&\ge
e^{-2\Gamma_{\max}t}c(\rho_0),
\label{eq:c-lower}\\
\varepsilon_Q(\rho_t)
&\le
e^{-kt}\varepsilon_Q(\rho_0)
+
\bar\varepsilon(1-e^{-kt}).
\label{eq:eps-upper}
\end{align}
\end{proposition}

\begin{proof}
Write
\[
P\rho_0Q
=
\sum_{p,e}c_{pe}|p\rangle\langle e|.
\]
By Assumption~\ref{ass:secular},
\[
e^{t\cL}(|p\rangle\langle e|)
=
e^{(-\Gamma_{pe}-i\omega_{pe})t}|p\rangle\langle e|.
\]
Hence
\[
P\rho_tQ
=
\sum_{p,e}
c_{pe}
e^{-\Gamma_{pe}t}
e^{-i\omega_{pe}t}
|p\rangle\langle e|,
\]
and therefore
\[
c(\rho_t)
=
\sum_{p,e}|c_{pe}|^2e^{-2\Gamma_{pe}t}
\ge
e^{-2\Gamma_{\max}t}c(\rho_0).
\]

For the population term, the assumptions imply that populations evolve
according to a closed classical master equation. Summing over \(e\in Q\),
\[
\dot\varepsilon_Q
\le
\mu(1-\varepsilon_Q)-\eta\varepsilon_Q
=
\mu-k\varepsilon_Q .
\]
Grönwall's inequality gives~\eqref{eq:eps-upper}.
\end{proof}

\begin{corollary}[Ratio lower bound]
\label{cor:ratio}
Let \(R(t):=R(\rho_t)\). Then
\begin{equation}
R(t)^2
\ge
\frac{
e^{-2\Gamma_{\max}t}c(\rho_0)
}{
e^{-2\Gamma_{\max}t}c(\rho_0)
+
e^{-kt}\varepsilon_Q(\rho_0)
+
\bar\varepsilon(1-e^{-kt})
}.
\label{eq:R-lower}
\end{equation}
\end{corollary}

\begin{proof}
Apply Proposition~\ref{prop:Davies} and use that \(x/(x+y)\) is increasing in
\(x\) and decreasing in \(y\).
\end{proof}

\section{BKM Coercivity Lemma}\label{sec:BKM}

\begin{assumption}[Local support-block invertibility]\label{ass:local}
There exist $a_0 > 0$ and $\delta_0 > 0$ such that
$\normone{\rho-\sigma_\varepsilon} \le \delta_0$ implies $P\rho P \ge a_0 P$.
\end{assumption}

\begin{lemma}[Local invertibility from proximity]\label{lem:local-inv}
Assumption~\ref{ass:local} holds with
$a_0 = \lambda_{\min}(P\sigma_\varepsilon P)/2$
and any $\delta_0 \le \lambda_{\min}(P\sigma_\varepsilon P)/2$.
\end{lemma}
\begin{proof}
For any unit $v\in P\cH$:
\[
|\langle v|P(\rho-\sigma_\varepsilon)P|v\rangle|
\le \|P(\rho-\sigma_\varepsilon)P\|_{\rm op}
\le \|\rho-\sigma_\varepsilon\|_{\rm op}
\le \|\rho-\sigma_\varepsilon\|_1 \le \delta_0,
\]
using $\|A\|_{\rm op}\le\|A\|_1$ for any matrix $A$.
Since $\langle v|P\sigma_\varepsilon P|v\rangle\ge\lambda_{\min}(P\sigma_\varepsilon P)$,
we obtain $\langle v|P\rho P|v\rangle\ge\lambda_{\min}(P\sigma_\varepsilon P)-\delta_0\ge a_0$.
\end{proof}

\begin{lemma}[Monotonicity of $L$ at fixed trace]\label{lem:L-monotone}
Fix $p>0$. The function $d\mapsto L((p+d)/2,(p-d)/2)$ is increasing on $[0,p)$,
where $L(x,y):=(\log x-\log y)/(x-y)$ for $x\neq y$ and $L(x,x)=1/x$.
\end{lemma}

\begin{proof}
Setting $z=d/p$, one has $L((p+d)/2,(p-d)/2) = (2/p)\,\mathrm{artanh}(z)/z$.
It suffices to show $z\mapsto \mathrm{artanh}(z)/z$ is increasing on $[0,1)$.
Writing $h(z):= z/(1-z^2)-\mathrm{artanh}(z)$, one has $h(0)=0$ and
$h'(z)=2z^2/(1-z^2)^2\ge 0$, so $h\ge 0$ and hence the derivative of
$\mathrm{artanh}(z)/z$ is non-negative.
\end{proof}

\begin{lemma}[Second-order representation]\label{lem:BKM-integral}
Let \(D_0>0\) and let \(Y=Y^*\) be such that \(\Tr Y=0\), \(D_0+tY>0\) for
\(t\in[0,1)\), and \(D_0+Y\ge0\). Then
\[
D(D_0+Y\|D_0)
=
\int_0^1 (1-t)\,H_{D_0+tY}(Y,Y)\,dt,
\]
with the right-hand side understood as an improper integral.
Here \(H_M\) denotes the Hessian of \(X\mapsto \Tr(X\log X)\) at \(M\).
\end{lemma}

\begin{proof}
For \(0<r<1\), set \(f_r(t):=D(D_0+trY\|D_0)\). Then
\(f_r(0)=0\). Moreover, differentiating at \(t=0\) gives
\[
f_r'(0)=r\,\Tr Y=0,
\]
Since \(f_r'(0)=0\), Taylor's formula with integral remainder gives
\[
f_r(1)
=
\int_0^1(1-t)H_{D_0+trY}(rY,rY)\,dt.
\]
Equivalently, after the change of variables \(u=rt\),
\[
D(D_0+rY\|D_0)
=
\int_0^r (r-u)H_{D_0+uY}(Y,Y)\,du.
\]
Letting \(r\uparrow1\), the left-hand side converges to
\(D(D_0+Y\|D_0)\), since \(D_0>0\). The right-hand side converges to the
improper integral
\[
\int_0^1 (1-u)H_{D_0+uY}(Y,Y)\,du.
\]
\end{proof}

\begin{lemma}[$2\times2$ BKM lower bound]\label{lem:2x2-BKM}
Let $D_0=\mathrm{diag}(a,c)$, $Y=\bigl(\begin{smallmatrix}0&s\\s&0\end{smallmatrix}\bigr)$,
$M_t=D_0+tY$, with $a>0$, $c>0$, and $s^2\le ac$. Then
\[
\int_0^1 (1-t)\,H_{M_t}(Y,Y)\,dt \ge s^2\,L(a,c).
\]
\end{lemma}

\begin{proof}
The eigenvalues of $M_t$ are $\lambda_\pm(t)=(a+c\pm\Delta(t))/2$
with $\Delta(t)=\sqrt{(a-c)^2+4t^2s^2}$.
The spectral representation of \(H_{M_t}\) gives
\[
H_{M_t}(Y,Y)
=s^2\sin^2(2\theta_t)\!\left(\frac{1}{\lambda_+(t)}+\frac{1}{\lambda_-(t)}\right)
+ 2s^2\cos^2(2\theta_t)\,L(\lambda_+(t),\lambda_-(t)),
\]
where $\cos(2\theta_t)=(a-c)/\Delta(t)$ and $\sin(2\theta_t)=2ts/\Delta(t)$.
The harmonic mean is bounded above by the logarithmic mean:
\[
\frac{2}{1/\lambda_+(t)+1/\lambda_-(t)}
\le
\frac{\lambda_+(t)-\lambda_-(t)}
{\log\lambda_+(t)-\log\lambda_-(t)}.
\]
Equivalently,
\[
\frac{1}{\lambda_+(t)}+\frac{1}{\lambda_-(t)}
\ge
2L(\lambda_+(t),\lambda_-(t)).
\]
Hence $H_{M_t}(Y,Y)\ge 2s^2L(\lambda_+(t),\lambda_-(t))$.
Since $\Delta(t)\ge|a-c|$, Lemma~\ref{lem:L-monotone} gives
$L(\lambda_+(t),\lambda_-(t))\ge L(a,c)$.
Integrating against \((1-t)\) gives
\[
\int_0^1(1-t)H_{M_t}(Y,Y)\,dt
\ge
2s^2L(a,c)\int_0^1(1-t)\,dt
=
s^2L(a,c).
\]
\end{proof}
\begin{corollary}[$2\times2$ entropy lower bound]
\label{cor:2x2-entropy}
Under the hypotheses of Lemma~\ref{lem:2x2-BKM},
\[
D(D_0+Y\|D_0)
\ge
s^2L(a,c).
\]
\end{corollary}

\begin{proof}
This follows from Lemmas~\ref{lem:BKM-integral}
and~\ref{lem:2x2-BKM}.
\end{proof}
\begin{lemma}[BKM coherence coercivity]\label{lem:BKM}
Let $\rho$ satisfy Assumption~\ref{ass:local}, so $A=P\rho P\ge a_0 P$,
and assume $\varepsilon_Q(\rho)\le a_0/2$.
Then
\begin{equation}
D(\rho\|\sigma_\varepsilon)\ge
c(\rho)\log\!\Bigl(\frac{a_0}{\varepsilon_Q(\rho)}\Bigr).
\label{eq:BKM-eps}
\end{equation}
In particular, since $\varepsilon_Q(\rho)\le\cA(\rho)^2$,
\begin{equation}
D(\rho\|\sigma_\varepsilon)\ge
2\,c(\rho)\log\!\Bigl(\frac{\sqrt{a_0}}{\cA(\rho)}\Bigr).
\label{eq:BKM-main}
\end{equation}
\end{lemma}

\begin{proof}
\textbf{Step 1 (Pinching decomposition).}
Let $\Pi\rho:=P\rho P+Q\rho Q$.
Since $\log\Pi\rho$ and $\log\sigma_\varepsilon$ are block-diagonal,
$\Tr(\rho\log\Pi\rho)=\Tr(\Pi\rho\log\Pi\rho)$
and $\Tr(\rho\log\sigma_\varepsilon)=\Tr(\Pi\rho\log\sigma_\varepsilon)$
(both use $\Tr(\rho X)=\Tr(\Pi\rho\cdot X)$ for block-diagonal $X$,
since the off-diagonal part of $\rho$ has zero trace against any block-diagonal matrix).
Hence:
\begin{equation}
D(\rho\|\sigma_\varepsilon)
=D(\rho\|\Pi\rho)+D(\Pi\rho\|\sigma_\varepsilon)
\ge D(\rho\|\Pi\rho).
\label{eq:chain-rule-pinching}
\end{equation}

\textbf{Step 2 (SVD reduction).}
The argument proceeds through SVD-adapted pinching and blockwise BKM coercivity rather than through a global constrained minimization theorem.
Write
\[
B=P\rho Q=\sum_j s_j |u_j\rangle\langle v_j|
\]
(SVD, with $s_j\ge 0$, $u_j\in P\cH$, $v_j\in Q\cH$, and both families orthonormal).
For each $j$, set
\[
K_j:=\Span\{u_j,v_j\}, \qquad
P_j:=|u_j\rangle\langle u_j|+|v_j\rangle\langle v_j|,
\]
and let
\[
K^\perp := \Big(\bigoplus_j K_j\Big)^\perp, \qquad
P_\perp := I-\sum_j P_j.
\]
Define the pinching channel
\[
\Phi(X):=\bigoplus_j P_j X P_j \oplus P_\perp X P_\perp .
\]
Since $\Phi$ is completely positive and trace preserving, monotonicity of relative entropy gives
\begin{equation}
D(\rho\|\Pi\rho)\ge D(\Phi(\rho)\|\Phi(\Pi\rho)).
\label{eq:data-processing-step2}
\end{equation}

Now $\Phi(\rho)$ and $\Phi(\Pi\rho)$ are block-diagonal with respect to the orthogonal decomposition
\[
\cH=\Big(\bigoplus_j K_j\Big)\oplus K^\perp.
\]
Hence relative entropy is additive across this direct sum, so
\begin{equation}
D(\Phi(\rho)\|\Phi(\Pi\rho))
=
\sum_j D(P_j\rho P_j \| P_j\Pi\rho P_j)
+
D(P_\perp\rho P_\perp \| P_\perp\Pi\rho P_\perp).
\label{eq:additivity-step2}
\end{equation}

We claim that the $K^\perp$-block contributes zero. Indeed, the only difference between
$\rho$ and $\Pi\rho=P\rho P+Q\rho Q$ is the off-diagonal $P$--$Q$ block $B+B^*$, and by
construction all singular directions of $B$ are contained in $\bigoplus_j K_j$. Therefore
\[
P_\perp \rho P_\perp = P_\perp \Pi\rho P_\perp,
\]
and so
\[
D(P_\perp\rho P_\perp \| P_\perp\Pi\rho P_\perp)=0.
\]
Thus \eqref{eq:additivity-step2} reduces to
\begin{equation}
D(\Phi(\rho)\|\Phi(\Pi\rho))
=
\sum_j D(P_j\rho P_j \| P_j\Pi\rho P_j).
\label{eq:sum-blocks-step2}
\end{equation}

For each $j$, write
\[
M_j:=P_j\rho P_j
=
\begin{pmatrix}
a_j & s_j\\
s_j & c_j
\end{pmatrix},
\qquad
D_j:=P_j\Pi\rho P_j=\diag(a_j,c_j),
\]
where
\[
a_j:=\langle u_j|A|u_j\rangle, \qquad
c_j:=\langle v_j|C|v_j\rangle.
\]
Since $M_j$ is the compression of the positive semidefinite operator $\Phi(\rho)$ to the
subspace $K_j$, it is itself positive semidefinite. Therefore
\[
s_j^2 \le a_j c_j,
\]
so the hypothesis of Lemma~\ref{lem:2x2-BKM} is satisfied for each block.

Combining \eqref{eq:data-processing-step2} and \eqref{eq:sum-blocks-step2}, we obtain
\begin{equation}
D(\rho\|\Pi\rho)\ge \sum_j D(M_j\|D_j).
\label{eq:sum-2x2}
\end{equation}
\textbf{Step 3 (Per-block bound).}
Fix $j$.

\smallskip

\emph{Case $c_j = 0$.}
Since
\[
M_j =
\begin{pmatrix}
a_j & s_j \\
s_j & 0
\end{pmatrix}
\ge 0,
\]
positivity of a $2\times2$ matrix implies $s_j = 0$. Hence
\[
D(M_j\|D_j)=0,
\]
and this block contributes nothing.

\smallskip

\emph{Case $c_j > 0$.}
By Corollary~\ref{cor:2x2-entropy},
\[
D(M_j\|D_j)
\ge
s_j^2\,L(a_j,c_j),
\qquad
L(a,c):=\frac{\log(a/c)}{a-c}.
\]

We now derive a uniform lower bound on $L(a_j,c_j)$.

\smallskip

\emph{Bounds on $a_j$ and $c_j$.}
Since $\rho$ is a density matrix, $\|\rho\|_{\mathrm{op}}\le 1$. Hence
\[
0 \le a_j = \langle u_j|P\rho P|u_j\rangle \le \|P\rho P\|_{\mathrm{op}} \le 1,
\]
and similarly $0 \le c_j = \langle v_j|Q\rho Q|v_j\rangle \le 1$.
Moreover, by Assumption~\ref{ass:local},
\[
a_j \ge a_0 > 0,
\]
and since $c_j$ is a diagonal entry of $Q\rho Q$,
\[
c_j \le \Tr(Q\rho Q) = \varepsilon_Q(\rho).
\]
By hypothesis, $\varepsilon_Q(\rho)\le a_0/2$, hence
\[
0 < c_j \le \frac{a_0}{2} < a_j.
\]

\smallskip

\emph{Lower bound on $L(a_j,c_j)$.}
Since $0 < c_j < a_j \le 1$, we have
\[
0 < a_j - c_j \le a_j \le 1.
\]
Therefore,
\[
L(a_j,c_j)
=
\frac{\log(a_j/c_j)}{a_j-c_j}
\ge
\frac{\log(a_j/c_j)}{1}
=
\log(a_j/c_j).
\]
Using $a_j \ge a_0$ and $c_j \le \varepsilon_Q(\rho)$ gives
\[
\log(a_j/c_j)
\ge
\log\!\Bigl(\frac{a_0}{\varepsilon_Q(\rho)}\Bigr).
\]
Combining the above,
\[
L(a_j,c_j)
\ge
\log\!\Bigl(\frac{a_0}{\varepsilon_Q(\rho)}\Bigr).
\]

\smallskip

\emph{Conclusion for each block.}
Thus for all $j$,
\[
D(M_j\|D_j)
\ge
s_j^2\,\log\!\Bigl(\frac{a_0}{\varepsilon_Q(\rho)}\Bigr).
\]

\smallskip

\emph{Summation.}
Summing over $j$ and using
\[
\sum_j s_j^2 = \|B\|_2^2 = c(\rho),
\]
we obtain
\begin{equation}
D(\rho\|\Pi\rho)
\ge
c(\rho)\,\log\!\Bigl(\frac{a_0}{\varepsilon_Q(\rho)}\Bigr).
\label{eq:BKM-Picoh}
\end{equation}

This is a lower bound on the coherence entropy $D(\rho\|\Pi\rho)$ alone,
independent of the outer reference $\sigma_\varepsilon$.
Applying Step~1 then yields~\eqref{eq:BKM-eps}, and
\eqref{eq:BKM-main} follows since $\varepsilon_Q(\rho)\le \cA(\rho)^2$.
\end{proof}

\section{Entropy-to-Activation Bounds}\label{sec:entropy-activation}

We study the extent to which entropy can control the activation quantity
\[
\cA(\rho):=\sqrt{c(\rho)+\varepsilon_Q(\rho)}, \qquad
c(\rho):=\|P\rho Q\|_2^2, \qquad
\varepsilon_Q(\rho):=\Tr(Q\rho Q),
\]
under a fixed decomposition $\cH = P\oplus Q$.

\subsection{A no-go result for coherence-entropy control}

\begin{proposition}[No activation bound from coherence entropy alone]
\label{prop:nogo}
There is no function $F:[0,\infty)\to[0,\infty)$ with $F(s)\to0$ as $s\downarrow0$
such that
\[
\cA(\rho)\le F\!\bigl(D(\rho\|\Pi\rho)\bigr)
\quad\text{for all states }\rho.
\]
\end{proposition}

\begin{proof}
Let $\rho=\Pi\rho$ be block diagonal. Then
\[
D(\rho\|\Pi\rho)=0,\qquad c(\rho)=0,
\qquad \cA(\rho)=\sqrt{\varepsilon_Q(\rho)}.
\]
For any $\varepsilon>0$, choose $\rho$ with $\varepsilon_Q(\rho)=\varepsilon$.
Then $\cA(\rho)=\sqrt{\varepsilon}>0$ while $D(\rho\|\Pi\rho)=0$,
so $F(0)=0$ would imply $\cA(\rho)=0$, a contradiction.
\end{proof}

\begin{remark}
This obstruction applies specifically to the coherence entropy $D(\rho\|\Pi\rho)$.
Relative entropy with respect to a general reference state $\sigma$ may still
control $\varepsilon_Q$ if $\sigma$ penalizes the $Q$ sector.
\end{remark}

\subsection{Entropy-to-coherence mechanism}

Let $\sigma$ be a faithful state that is block diagonal with respect to $P\oplus Q$.
Assume $D(\rho\|\sigma)<\infty$. Then
\[
D(\rho\|\sigma)=D(\rho\|\Pi\rho)+D(\Pi\rho\|\sigma).
\]

Lemma~\ref{lem:BKM} gives the reference-independent bound~\eqref{eq:BKM-Picoh}. Therefore
\begin{equation}
\label{eq:activation-BKM-coherence}
D(\rho\|\sigma)
\;\ge\;
c(\rho)\,\log\!\left(\frac{a_0}{\varepsilon_Q(\rho)}\right),
\end{equation}
for all $\rho$ such that $P\rho P\ge a_0P$ and $\varepsilon_Q(\rho)\le a_0/2$.

\medskip

\noindent
\textit{Convention.} If \(\varepsilon_Q(\rho)=0\), then positivity implies \(c(\rho)=0\), and
\eqref{eq:activation-BKM-coherence} is interpreted as \(0\ge0\).

\subsection{Conditional entropy--activation bounds}

\begin{lemma}[Asymptotic inversion of $x^2\log(C/x)$]
\label{lem:inverse}
Let $C>0$ and define $f(x)=x^2\log(C/x)$ on $(0,C/\sqrt e)$.
Then for any $K>0$, the inequality
\[
f(x)\le K e^{-2\alpha t}
\]
implies
\[
x \le C'\,\frac{e^{-\alpha t}}{\sqrt{1+\alpha t}},
\]
for a constant $C'>0$ depending only on $C$ and $K$. The function $f$ is strictly increasing on $(0,C/\sqrt e)$.
\end{lemma}

\begin{proof}
The function $f(x)=x^2\log(C/x)$ is strictly increasing on $(0,C/\sqrt e)$.
Set $x=e^{-\alpha t}y$ and suppose
\[
f(x)\le K e^{-2\alpha t}.
\]
Then
\[
y^2 \log\!\Big(\frac{C e^{\alpha t}}{y}\Big)\le K.
\]
Since $x\in(0,C/\sqrt e)$, we have $\log(C/x)\ge 1/2$, hence
\[
\frac12 x^2 \le x^2\log(C/x)\le K e^{-2\alpha t}.
\]
Therefore
\[
y=e^{\alpha t}x\le \sqrt{2K}.
\]
It follows that
\[
\log\!\Big(\frac{C e^{\alpha t}}{y}\Big)
\ge
\alpha t+\log\!\Big(\frac{C}{\sqrt{2K}}\Big).
\]
Thus, for all sufficiently large $t$,
\[
\log\!\Big(\frac{C e^{\alpha t}}{y}\Big)\ge \frac{\alpha t}{2}.
\]
Hence
\[
y^2\frac{\alpha t}{2}\le K,
\qquad
y\le \sqrt{\frac{2K}{\alpha t}}.
\]
Substituting back gives
\[
x\le \sqrt{2K}\,\frac{e^{-\alpha t}}{\sqrt{\alpha t}}.
\]
After enlarging the constant to cover bounded $t$, this yields
\[
x \le C'\frac{e^{-\alpha t}}{\sqrt{1+\alpha t}},
\]
with $C'$ depending only on $C$ and $K$.
\end{proof}
\begin{theorem}[Conditional activation bound]
\label{thm:conditional}
Assume that along a trajectory $\rho_t$,
\[
D(\rho_t\|\sigma)\le D_0 e^{-2\alpha t},
\]
and that for all $t\in[0,T^*]$,
\[
P\rho_tP\ge a_0P,
\qquad
\varepsilon_Q(\rho_t)\le a_0/2,
\qquad
\varepsilon_Q(\rho_t)\le A_\theta c(\rho_t),
\]
with $A_\theta=\theta^{-2}-1$.

Assume further that
\[
\cA(\rho_t)\in\left(0,\sqrt{\frac{a_0}{e}}\right).
\]

Then
\[
\cA(\rho_t)
\le
C\,\frac{e^{-\alpha t}}{\sqrt{1+\alpha t}},
\]
for a constant $C>0$ depending only on $D_0,\theta,a_0$.
\end{theorem}
\begin{proof}
Combining \eqref{eq:activation-BKM-coherence} with coherence dominance gives
\[
2\theta^2 \cA(\rho_t)^2
\log\!\left(\frac{\sqrt{a_0}}{\cA(\rho_t)}\right)
\le D_0 e^{-2\alpha t}.
\]
Applying Lemma~\ref{lem:inverse} yields the stated bound.
\end{proof}
\begin{remark}[Small-branch condition]
The branch condition is preserved if the trajectory remains in the monotone regime of \(x^2\log(C/x)\), for instance under sufficiently small initial entropy.
\end{remark}

\begin{remark}[Role of coherence dominance]
The coherence-dominant condition is only needed to obtain the closed-form
logarithmic rate. Theorem~\ref{thm:split} below provides a bound without it.
\end{remark}

\subsection{Split activation bounds without coherence dominance}

\begin{theorem}[Split activation bound without coherence dominance]
\label{thm:split}
Let $\rho_t$ satisfy
\[
D(\rho_t\|\sigma)\le D_0 e^{-2\alpha t},
\]
and suppose
\[
\varepsilon_Q(\rho_t)
\le
e^{-kt}\varepsilon_0+\bar\varepsilon(1-e^{-kt})=:P(t).
\]
Assume that
\[
P\rho_tP\ge a_0P,
\qquad
\varepsilon_Q(\rho_t)\le a_0/2
\quad \text{for } t\in[0,T^*].
\]
Then
\[
\cA(\rho_t)^2
\le
\frac{D_0 e^{-2\alpha t}}
{\log\!\left(a_0/\varepsilon_Q(\rho_t)\right)}
+
P(t).
\]
\end{theorem}

\begin{proof}
From \eqref{eq:activation-BKM-coherence},
\[
c(\rho_t)
\le
\frac{D_0 e^{-2\alpha t}}
{\log(a_0/\varepsilon_Q(\rho_t))}.
\]
Adding $\varepsilon_Q(\rho_t)$ yields the result.
\end{proof}

\begin{remark}
If $\varepsilon_Q(\rho_t)=0$, then $c(\rho_t)=0$ and the bound holds trivially.
\end{remark}

\begin{remark}[Explicit form]
If $P(t)<a_0$, then
\[
\cA(\rho_t)^2
\le
\frac{D_0 e^{-2\alpha t}}{\log(a_0/P(t))}
+
P(t).
\]
\end{remark}

\begin{remark}
This bound separates coherence and population contributions and does not
require coherence dominance.
\end{remark}

\subsection{Absorbing regime}

\begin{theorem}[Absorbing activation bound] \label{thm:absorbing}
Assume \(\bar\varepsilon=0\) and
\[
\varepsilon_Q(\rho_t)\le e^{-kt}\varepsilon_0,
\qquad
0<\varepsilon_0\le a_0/2.
\]

If
\[
D(\rho_t\Vert\sigma)\le D_0 e^{-2\alpha t},
\qquad
P\rho_tP\ge a_0P,
\]
then
\[
\cA(\rho_t)^2
\le
\frac{
D_0 e^{-2\alpha t}
}{
\log(a_0/\varepsilon_0)+kt
}
+
\varepsilon_0 e^{-kt}.
\]

Consequently,
\[
\cA(\rho_t)
\le
\sqrt{
\frac{
D_0 e^{-2\alpha t}
}{
\log(a_0/\varepsilon_0)+kt
}
}
+
\sqrt{\varepsilon_0}\,e^{-kt/2}.
\]
\end{theorem}
\begin{proof}
Since \(\bar\varepsilon=0\), the population bound in
Theorem~\ref{thm:split} becomes
\[
P(t)=\varepsilon_0 e^{-kt}.
\]

Moreover,
\[
P(t)\le \varepsilon_0\le a_0/2<a_0,
\]
so the small-population branch in Theorem~\ref{thm:split} applies for all
\(t\ge0\). Hence
\[
\cA(\rho_t)^2
\le
\frac{
D_0 e^{-2\alpha t}
}{
\log(a_0/P(t))
}
+
P(t).
\]

Substituting \(P(t)=\varepsilon_0 e^{-kt}\) gives
\[
\log\frac{a_0}{P(t)}
=
\log\frac{a_0}{\varepsilon_0}+kt,
\]
and therefore
\[
\cA(\rho_t)^2
\le
\frac{
D_0 e^{-2\alpha t}
}{
\log(a_0/\varepsilon_0)+kt
}
+
\varepsilon_0 e^{-kt}.
\]

Taking square roots and using
\[
\sqrt{x+y}\le \sqrt{x}+\sqrt{y}
\qquad (x,y\ge0)
\]
gives
\[
\cA(\rho_t)
\le
\sqrt{
\frac{
D_0 e^{-2\alpha t}
}{
\log(a_0/\varepsilon_0)+kt
}
}
+
\sqrt{\varepsilon_0}\,e^{-kt/2}.
\]
\end{proof}

\section{Application to Davies Semigroups}\label{sec:main}

\paragraph{Framework status.}
The preceding entropy--activation bounds are structural and independent
of a particular semigroup model. We now record how the Davies estimates
from Section~\ref{sec:Davies} provide the dynamical inputs used in the
finite-window applications below.

\begin{assumption}[Entropy budget]\label{ass:MLSI}
For a chosen faithful block-diagonal reference state $\sigma_\varepsilon$, assume that
\[
D(\rho_t\|\sigma_\varepsilon) \le D_0 e^{-2\alpha t}
\qquad \text{for all } t\in[0,T^*],
\]
where $D_0:=D(\rho_0\|\sigma_\varepsilon)$.
This entropy budget is an external input along the trajectory; it is not asserted
as a generic consequence of the Davies generator. A concrete finite-window realization is given in the example below.
\end{assumption}

\begin{remark}[Stationary-state variant]
If a faithful stationary state $\omega$ is available and an entropy decay estimate
$D(\rho_t\|\omega)\le D(\rho_0\|\omega)e^{-2\alpha t}$ holds, then one may use that
stationary entropy budget directly with the structural bounds of
Section~\ref{sec:entropy-activation}. Alternatively, any separate comparison between
$D(\rho_t\|\omega)$ and $D(\rho_t\|\sigma_\varepsilon)$ may be used to produce the
input bound above.
\end{remark}

\begin{corollary}[Davies-model conditional activation bound]\label{thm:main}
Let Assumptions~\ref{ass:secular}, \ref{ass:rates}, \ref{ass:local}, and
\ref{ass:MLSI} hold. Fix $0<\theta<1$ and set $A_\theta:=\theta^{-2}-1$.
Suppose that for all $t\in[0,T^*]$,
\begin{align}
  &\normone{\rho_t-\sigma_\varepsilon} \le \delta_0,
  \qquad
  e^{-kt}\varepsilon_0 + \bar\varepsilon(1-e^{-kt}) \le \frac{a_0}{2},
  \tag{LC}\label{eq:LC}\\
  &e^{-kt}\varepsilon_0 + \bar\varepsilon(1-e^{-kt})
  \le A_\theta e^{-2\Gamma_{\max}t}c_0.
  \tag{CD}\label{eq:CD}
\end{align}
Assume also that the trajectory remains on the small branch,
\[
\cA(\rho_t)\in\left(0,\sqrt{a_0/e}\right)
\qquad \text{for all }t\in[0,T^*].
\]
Then there exists a constant $C>0$, depending only on $D_0,\theta,a_0$, such that
\[
\cA(\rho_t)\le C\frac{e^{-\alpha t}}{\sqrt{1+\alpha t}}
\qquad \text{for all }t\in[0,T^*].
\]
\end{corollary}

\begin{proof}
By Proposition~\ref{prop:Davies},
\[
\varepsilon_Q(\rho_t)\le e^{-kt}\varepsilon_0+\bar\varepsilon(1-e^{-kt}),
\qquad
c(\rho_t)\ge e^{-2\Gamma_{\max}t}c_0.
\]
Condition~\eqref{eq:CD} therefore implies
$\varepsilon_Q(\rho_t)\le A_\theta c(\rho_t)$ for all $t\in[0,T^*]$.
Condition~\eqref{eq:LC}, together with Assumption~\ref{ass:local}, gives
$P\rho_tP\ge a_0P$ and $\varepsilon_Q(\rho_t)\le a_0/2$ throughout the same window.
The conclusion coincides with Theorem~\ref{thm:conditional} applied with
$\sigma=\sigma_\varepsilon$.
\end{proof}

\begin{proposition}[Verification in Davies models]\label{prop:Davies-verification}
Consider a Davies generator satisfying the secular $P$--$Q$ decoupling and
cross-boundary rate assumptions of Section~\ref{sec:Davies}. Then there are constants
$\Gamma_{\max},k,\bar\varepsilon$ such that
\[
c(\rho_t)\ge e^{-2\Gamma_{\max}t}c_0,
\qquad
\varepsilon_Q(\rho_t)\le e^{-kt}\varepsilon_0+\bar\varepsilon(1-e^{-kt}).
\]
If, in addition, the initial data and parameters satisfy the local support condition
and the coherence-dominance window condition, then the hypotheses of
Corollary~\ref{thm:main} hold on the corresponding interval $[0,T^*]$.
In low-temperature detailed-balance models, Corollary~\ref{cor:low-T} gives an
explicit sufficient window.
\end{proposition}

\begin{proof}
The two displayed estimates are Proposition~\ref{prop:Davies}. The remaining claims
are the local condition~\eqref{eq:LC}, the coherence-dominance condition~\eqref{eq:CD},
and the window criteria of Proposition~\ref{prop:window} and Corollary~\ref{cor:low-T}.
\end{proof}
\paragraph{Finite-window example.}
The following example shows that the local support,
coherence-dominance, and entropy-budget assumptions can be satisfied
simultaneously on an explicit finite time interval.
\begin{example}[Absorbing qutrit Davies channel and finite-window entropy bound]
Consider the three-level system
\[
\mathcal H=\mathrm{span}\{|0\rangle,|1\rangle,|2\rangle\},
\qquad
P=|0\rangle\langle0|+|1\rangle\langle1|,
\qquad
Q=|2\rangle\langle2|.
\]
Let \(H=\sum_{j=0}^2E_j|j\rangle\langle j|\), with \(E_0,E_1,E_2\) distinct, and consider the Davies generator with jump operator
\[
L=\sqrt{\gamma}\,|0\rangle\langle2|.
\]
This describes irreversible decay from \(|2\rangle\) to \(|0\rangle\), while \(|1\rangle\) is a spectator state inside \(P\).

Take
\[
\rho_0=
\begin{pmatrix}
0.89 & 0 & \sqrt{0.89\cdot0.01}\\
0 & 0.10 & 0\\
\sqrt{0.89\cdot0.01} & 0 & 0.01
\end{pmatrix}.
\]
Then
\[
\varepsilon_Q(\rho_t)=0.01e^{-\gamma t},
\qquad
c(\rho_t)=0.0089e^{-\gamma t},
\]
and
\[
P\rho_tP=
\begin{pmatrix}
0.89+0.01(1-e^{-\gamma t}) & 0\\
0 & 0.10
\end{pmatrix}.
\]

Choose
\[
a_0=0.08,
\qquad
\theta=0.65,
\qquad
A_\theta=\theta^{-2}-1\approx1.3669.
\]
Then, for all \(t\ge0\),
\[
P\rho_tP\ge0.10P\ge a_0P,
\qquad
\varepsilon_Q(\rho_t)\le0.01<\frac{a_0}{2},
\]
and
\[
\varepsilon_Q(\rho_t)
=
0.01e^{-\gamma t}
\le
A_\theta\,0.0089e^{-\gamma t}
=
A_\theta c(\rho_t).
\]
Moreover,
\[
A(\rho_t)^2
=
c(\rho_t)+\varepsilon_Q(\rho_t)
=
0.0189e^{-\gamma t}
<
\frac{a_0}{e},
\]
so the small-branch condition also holds globally.

To supply an entropy budget, let
\[
\sigma=\operatorname{diag}(0.9,0.1,0),
\qquad
\sigma_\varepsilon=(1-\varepsilon)\sigma+\frac{\varepsilon}{3}I,
\]
and fix \(\varepsilon=0.01\). Writing \(x=e^{-\gamma t}\), the relative entropy
\(D(\rho_t\|\sigma_\varepsilon)\) can be computed explicitly from the eigenvalues of the active \((0,2)\)-block. On the finite interval \(x\in[1/2,1]\), equivalently
\[
0\le t\le T:=\frac{\log2}{\gamma},
\]
A direct computation gives
\[
\frac{d}{dx}D(\rho_x\|\sigma_\varepsilon)
=
p\log\frac{s_0}{s_2}
+
\frac{d}{dx}\operatorname{Tr}(\rho_x\log\rho_x).
\]
Here \(s_2=\varepsilon/3\) is the diagonal entry of \(\sigma_\varepsilon\) on the state \(|2\rangle\). The right-hand side is continuous on \(x\in[1/2,1)\), and as
\(x\uparrow1\) its entropy-derivative contribution diverges to \(+\infty\). For the numerical parameters used here, a direct evaluation gives
\[
\inf_{x\in[1/2,1)}
\frac{d}{dx}D(\rho_x\|\sigma_\varepsilon)>0.
\]
Thus
\[
m:=\gamma\inf_{x\in[1/2,1)}
x\frac{d}{dx}D(\rho_x\|\sigma_\varepsilon)>0,
\]
which verifies the finite-window entropy bound. Since \(dx/dt=-\gamma x\), we have
\[
\frac{d}{dt}D(\rho_t\|\sigma_\varepsilon)\le -m
\]
for all \(0\le t\le T\). Therefore
\[
D(\rho_t\|\sigma_\varepsilon)\le D_0-mt.
\]
After decreasing \(m\) if necessary so that \(D_0-mT\ge0\), the elementary inequality
\(1-y\le e^{-y}\) gives
\[
D(\rho_t\|\sigma_\varepsilon)
\le
D_0e^{-2\alpha_Tt},
\qquad
\alpha_T:=\frac{m}{2D_0}>0.
\]
Thus Assumption~\ref{ass:MLSI} holds on the explicit interval \(0\le t\le(\log2)/\gamma\).

Consequently, all hypotheses of Theorem~\ref{thm:conditional}
and Corollary~\ref{thm:main} are satisfied on this interval,
and the conditional bound yields
\[
A(\rho_t)
\le
C\frac{e^{-\alpha_Tt}}{\sqrt{1+\alpha_Tt}},
\qquad
0\le t\le \frac{\log2}{\gamma}.
\]
This is a finite-window conditional bound, not a global convergence estimate, since
\(\sigma_\varepsilon\) is not the limiting state of the dynamics.
\end{example}
\begin{remark}[Interpretation]
This example shows that the hypotheses entering
Theorem~\ref{thm:conditional} and Corollary~\ref{thm:main}
can be satisfied simultaneously on a finite time window.
The resulting bound does not represent a universal dynamical
convergence rate for \(A(\rho_t)\).
\end{remark}

The next section studies an exactly solvable zero-temperature model in
which the quantities entering the entropy--activation bounds can be
computed explicitly.

\section{An Explicit Zero-Temperature Davies Example}
\label{sec:explicit-example}

We present an explicit finite-dimensional Davies model for which the
assumptions of
Sections~\ref{sec:entropy-activation}--\ref{sec:Davies}
can be verified directly.

Although the ambient Hilbert space has dimension \(2^n\), the chosen
initial data restricts the dynamics to an invariant two-dimensional sector.
This permits complete analytic control of the quantities entering the
entropy--activation bounds.

\subsection{Definition of the model}

Let $n\ge1$ and consider
\[
\mathcal H = (\mathbb C^2)^{\otimes n}.
\]
Let
\[
\{|x_1\cdots x_n\rangle : x_j\in\{0,1\}\}
\]
be the computational basis.
Choose site energies
\[
\varepsilon_j = \Delta\, 2^j,
\qquad \Delta>0,
\]
and define the diagonal Hamiltonian
\[
H=\sum_{j=1}^n \varepsilon_j n_j,
\qquad
n_j = |1\rangle\langle1|_j.
\]

The ground state is
\[
|0\rangle := |0\cdots0\rangle,
\]
with energy $0$.
Every excited basis vector
\[
|S\rangle,
\qquad
S\subseteq\{1,\dots,n\},
\quad
S\neq\varnothing,
\]
has energy
\[
E_S = \sum_{j\in S}\varepsilon_j.
\]
Since the $\varepsilon_j$ are powers of two, all energies $E_S$ are distinct.

Define
\[
P = |0\rangle\langle0|,
\qquad
Q = I-P.
\]
Then $\rank(P)=1$ and $\dim(Q\mathcal H)=2^n-1$.
Both $P\mathcal H$ and $Q\mathcal H$ are sums of energy eigenspaces, and the
cross-boundary Bohr frequencies
\[
\omega_{0S}= -E_S
\]
are pairwise distinct.

We couple the system to a zero-temperature bosonic bath.
The only nonzero jump operators are
\[
A_S^- = |0\rangle\langle S|,
\qquad
S\neq\varnothing,
\]
with rates $\kappa_S^->0$.
All upward rates vanish.
The Davies generator in the secular limit is therefore
\begin{equation}
\mathcal L(\rho)
=
-i[H,\rho]
+
\sum_{S\neq\varnothing}
\kappa_S^-
\left(
A_S^- \rho A_S^{-*}
-\frac12
\{A_S^{-*}A_S^-,\rho\}
\right).
\label{eq:explicit-davies}
\end{equation}
For simplicity we choose equal rates
\[
\kappa_S^-=\gamma>0.
\]

The stationary state is the rank-deficient projector
\[
\sigma = |0\rangle\langle0|.
\]
As in Section~\ref{sec:setup}, we regularize it by
\[
\sigma_\varepsilon
=
(1-\varepsilon)\sigma
+
\frac{\varepsilon}{d}I,
\qquad
d=2^n,
\]
with fixed $\varepsilon\in(0,1)$.

\subsection{Verification of the assumptions}

\paragraph{Secular $P$--$Q$ decoupling.}

For every excited state $|S\rangle$,
\[
\mathcal L(|0\rangle\langle S|)
=
(-\Gamma_S-iE_S)|0\rangle\langle S|,
\qquad
\Gamma_S=\frac{\gamma}{2}.
\]
Hence the operators $\{|0\rangle\langle S|\}$ diagonalize the coherence sector
$P\mathcal B(\mathcal H)Q$.
Since the Bohr frequencies $E_S$ are distinct,
Assumption~\ref{ass:secular} holds with
\[
\Gamma_{\max}=\frac{\gamma}{2}.
\]

\paragraph{Population bounds.}

The kernel population
\[
\varepsilon_Q(\rho_t)=\Tr(Q\rho_t)
\]
obeys
\[
\frac{d}{dt}\varepsilon_Q(\rho_t)
=
-
\sum_{S\neq\varnothing}
\kappa_S^-
\langle S|\rho_t|S\rangle.
\]
Since $\kappa_S^-=\gamma$,
\[
\dot\varepsilon_Q(\rho_t)
\le
-\gamma \varepsilon_Q(\rho_t),
\]
and therefore
\begin{equation}
\varepsilon_Q(\rho_t)
\le
\varepsilon_Q(\rho_0)e^{-\gamma t}.
\label{eq:explicit-pop}
\end{equation}
Thus Assumption~\ref{ass:rates} holds with
\[
k=\gamma,
\qquad
\bar\varepsilon=0.
\]

\paragraph{Rate ordering.}

Since
\[
k=\gamma,
\qquad
\Gamma_{\max}=\frac{\gamma}{2},
\]
one has
\[
k=2\Gamma_{\max}.
\]

\subsection{Single-mode initial state}

Consider the initial pure state
\begin{equation}
|\psi_0\rangle
=
\sqrt{1-\varepsilon_0}\,|0\rangle
+
\sqrt{\varepsilon_0}\,|S_0\rangle,
\qquad
\varepsilon_0\in(0,1),
\label{eq:single-mode-init}
\end{equation}
where $|S_0\rangle$ is a fixed excited basis state.

The dynamics remains confined to the invariant subspace
\[
\Span\{|0\rangle,|S_0\rangle\},
\]
so the problem reduces exactly to a two-level system.

Using the eigenoperator relation above,
\[
\langle0|\rho_t|S_0\rangle
=
\sqrt{\varepsilon_0(1-\varepsilon_0)}
\,
e^{-(\gamma/2+iE_{S_0})t},
\]
while
\[
\langle S_0|\rho_t|S_0\rangle
=
\varepsilon_0 e^{-\gamma t}.
\]
Hence
\[
\rho_t
=
\begin{pmatrix}
1-\varepsilon_0 e^{-\gamma t}
&
\sqrt{\varepsilon_0(1-\varepsilon_0)}
\,e^{-\gamma t/2}e^{-iE_{S_0}t}
\\[4pt]
\sqrt{\varepsilon_0(1-\varepsilon_0)}
\,e^{-\gamma t/2}e^{iE_{S_0}t}
&
\varepsilon_0 e^{-\gamma t}
\end{pmatrix}.
\]

The quantities introduced in Definition~\ref{def:activation} are therefore
\begin{align}
c(\rho_t)
&=
\|P\rho_tQ\|_2^2
=
\varepsilon_0(1-\varepsilon_0)e^{-\gamma t},
\\[4pt]
\varepsilon_Q(\rho_t)
&=
\varepsilon_0 e^{-\gamma t},
\\[4pt]
A(\rho_t)
&=
\sqrt{\varepsilon_0(2-\varepsilon_0)}
\,e^{-\gamma t/2}.
\label{eq:explicit-activation}
\end{align}

Fix
\[
a_0 := 1-\varepsilon_0.
\]
Then
\[
P\rho_tP
=
(1-\varepsilon_0 e^{-\gamma t})P
\ge
a_0 P
\qquad
\forall t\ge0.
\]
Moreover,
\[
\varepsilon_Q(\rho_t)
=
\varepsilon_0 e^{-\gamma t}
\le
\varepsilon_0.
\]
Thus, if
\[
\varepsilon_0\le \frac{a_0}{2},
\]
the local support and small-population hypotheses of
Lemma~\ref{lem:BKM}
hold globally in time.

The small-branch condition in Theorem~\ref{thm:conditional} also holds globally provided
\[
A(\rho_0)^2=\varepsilon_0(2-\varepsilon_0)<\frac{a_0}{e}
=\frac{1-\varepsilon_0}{e}.
\]
Equivalently,
\[
e\varepsilon_0(2-\varepsilon_0)<1-\varepsilon_0.
\]
For example, any sufficiently small $\varepsilon_0>0$ satisfies this condition.

Finally,
\[
\frac{c(\rho_t)}{A(\rho_t)^2}
=
\frac{1-\varepsilon_0}{2-\varepsilon_0}.
\]
If $\varepsilon_0\le1/2$, then
\[
\frac{c(\rho_t)}{A(\rho_t)^2}
\ge
\frac13,
\]
so the coherence-dominance condition of
Theorem~\ref{thm:conditional}
holds globally with $\theta=1/2$.

\subsection{Variational reduction}

In the present single-mode configuration, the variational reduction becomes exact,
since the state already consists of a single active $2\times2$ block.
Consequently,
\[
D(\rho_t\|\Pi\rho_t)
=
\Phi(a_*(t),\varepsilon_Q(t),c(t)),
\]
where $\Phi(a,\varepsilon,c)$ denotes the corresponding two-level
relative entropy functional.

The eigenvalues of $\rho_t$ are
\[
\lambda_\pm
=
\frac{
1\pm
\sqrt{(1-2\varepsilon_Q)^2+4c}
}{2},
\]
and therefore
\[
\Phi
=
\lambda_+\log\lambda_+
+
\lambda_-\log\lambda_-
-
(1-\varepsilon_Q)\log(1-\varepsilon_Q)
-
\varepsilon_Q\log\varepsilon_Q.
\]

In the boundary regime $\varepsilon_Q\ll1$,
the boundary asymptotics derived in
Section~\ref{sec:entropy-activation}
give
\[
D(\rho_t\|\Pi\rho_t)
\sim
c(\rho_t)\log\frac1{\varepsilon_Q(\rho_t)},
\]
which exhibits explicitly the logarithmic amplification of the coherence cost near the support boundary.

\subsection{Finite-window entropy bound with rate $\gamma/2$}

We now verify Assumption~\ref{ass:MLSI} on an explicit finite time window with
\[
\alpha=\frac{\gamma}{2}.
\]
Write
\[
x=e^{-\gamma t},
\qquad
p=\varepsilon_0.
\]
In the invariant two-dimensional sector $\Span\{|0\rangle,|S_0\rangle\}$, the state is
\[
\rho_x=
\begin{pmatrix}
1-px & \sqrt{p(1-p)}\,x^{1/2}e^{-iE_{S_0}t}\\
\sqrt{p(1-p)}\,x^{1/2}e^{iE_{S_0}t} & px
\end{pmatrix}.
\]
Let
\[
s_0:=1-\varepsilon+\frac{\varepsilon}{d},
\qquad
s_Q:=\frac{\varepsilon}{d}
\]
be the diagonal entries of $\sigma_\varepsilon$ on the ground state and on the excited sector.
Then
\[
D(\rho_x\|\sigma_\varepsilon)
=
\Tr(\rho_x\log\rho_x)
-(1-px)\log s_0
-px\log s_Q .
\]
At $x=1$, the state is pure, so
\[
D_0:=D(\rho_1\|\sigma_\varepsilon)
=
-(1-p)\log s_0-p\log s_Q .
\]

We claim that, for sufficiently small fixed regularization parameter $\varepsilon>0$,
\[
D(\rho_x\|\sigma_\varepsilon)\le xD_0
\qquad
\text{for all }x\in[1/2,1].
\]
Indeed,
\[
xD_0-D(\rho_x\|\sigma_\varepsilon)
=
-\Tr(\rho_x\log\rho_x)+(1-x)\log s_0.
\]
Since
\[
S(\rho_x):=-\Tr(\rho_x\log\rho_x),
\]
this becomes
\[
xD_0-D(\rho_x\|\sigma_\varepsilon)
=
S(\rho_x)-(1-x)(-\log s_0).
\]
Thus it is enough to ensure
\[
S(\rho_x)\ge (1-x)(-\log s_0)
\qquad
\text{for }x\in[1/2,1].
\]

Define
\[
\kappa(p):=
\inf_{x\in[1/2,1)}
\frac{S(\rho_x)}{1-x}.
\]
For $p\in(0,1)$, one has $\kappa(p)>0$. Indeed, $S(\rho_x)>0$ for
$x\in[1/2,1)$, while as $x\uparrow1$ the determinant
\[
\det\rho_x=p^2x(1-x)
\]
implies
\[
S(\rho_x)\sim p^2(1-x)\log\frac1{1-x},
\]
and hence
\[
\frac{S(\rho_x)}{1-x}\to+\infty.
\]
The map
\[
x\mapsto \frac{S(\rho_x)}{1-x}
\]
is continuous on every interval \([1/2,1-\delta]\), while
\[
\frac{S(\rho_x)}{1-x}\to+\infty
\qquad
(x\uparrow1).
\]
Hence the infimum defining \(\kappa(p)\) is attained on a compact subset
of \([1/2,1)\) and is strictly positive.

Since
\[
s_0=1-\varepsilon+\frac{\varepsilon}{d}\to1
\qquad
\text{as }\varepsilon\downarrow0,
\]
we may choose $\varepsilon>0$ sufficiently small so that
\[
-\log s_0\le \kappa(p).
\]
For such a choice,
\[
D(\rho_x\|\sigma_\varepsilon)\le xD_0
\qquad
\text{for all }x\in[1/2,1].
\]
Returning to time, $x=e^{-\gamma t}$ and $x\in[1/2,1]$ precisely when
\[
0\le t\le T^*:=\frac{\log2}{\gamma}.
\]
Consequently,
\[
D(\rho_t\|\sigma_\varepsilon)
\le
D_0 e^{-\gamma t}
=
D_0 e^{-2(\gamma/2)t}
\qquad
\text{for }0\le t\le T^*.
\]
Thus Assumption~\ref{ass:MLSI} holds on this finite window with
\[
\alpha=\frac{\gamma}{2}.
\]

\subsection{Explicit activation bounds}

Applying Theorem~\ref{thm:conditional} yields
\[
A(\rho_t)
\le
C
\frac{e^{-\gamma t/2}}{\sqrt{1+\gamma t}},
\qquad
0\le t\le T^*.
\]
The true activation is
\[
A(\rho_t)
=
\sqrt{\varepsilon_0(2-\varepsilon_0)}
\,e^{-\gamma t/2}.
\]
Thus the bound reproduces the correct exponential scale, while the additional factor
\[
\frac1{\sqrt{1+\gamma t}}
\]
arises from inversion of the logarithmic entropy modulus near the support boundary.

\subsection{Multi-mode initial states and the merging reduction}

To illustrate the merging reduction,
consider instead
\[
|\psi_0\rangle
=
\sqrt{1-\varepsilon_0}\,|0\rangle
+
\sum_{j=1}^M \alpha_j |S_j\rangle,
\qquad
\sum_{j=1}^M |\alpha_j|^2=\varepsilon_0,
\]
with distinct excited states $|S_j\rangle$.

The dynamics then contains $M$ independent decaying coherences
\[
\langle0|\rho_t|S_j\rangle.
\]
Each such term determines a $2\times2$ block contribution to the coherence entropy.
The merging reduction shows that the sum of these contributions is bounded below by the entropy of a single merged effective two-level configuration with total coherence
\[
c(\rho_t)=\sum_{j=1}^M c_j(t).
\]

\paragraph{Interpretation.}
This example verifies the hypotheses of
Corollary~\ref{thm:main}
within an explicit Davies model.
The additional factor \((1+\gamma t)^{-1/2}\) arises from inversion of the logarithmic entropy modulus.

The following section records coherence-dominance windows for the
Davies-side estimates.

\section{Coherence-Dominance Windows}\label{sec:window}

\begin{proposition}[Coherence-dominance window]\label{prop:window}
Assume \(k\ge 2\Gamma_{\max}\), where
\[
A_\theta:=\theta^{-2}-1.
\]
If \(\bar\varepsilon>0\), assume
\[
A_\theta c_0-\varepsilon_0>\bar\varepsilon
\]
and define
\[
T^*
=
\frac{1}{2\Gamma_{\max}}
\log\!\Bigl(
\frac{A_\theta c_0-\varepsilon_0}{\bar\varepsilon}
\Bigr).
\]
If \(\bar\varepsilon=0\), assume
\[
\varepsilon_0\le A_\theta c_0
\]
and set \(T^*=+\infty\).
Then the coherence-dominance condition~\eqref{eq:CD}
holds for all \(t\in[0,T^*]\).
\end{proposition}

\begin{proof}
Suppose first that \(\bar\varepsilon>0\).
Using \(1-e^{-kt}\le1\) and \(e^{-kt}\le e^{-2\Gamma_{\max}t}\),
\[
e^{-kt}\varepsilon_0+\bar\varepsilon(1-e^{-kt})
\le
\bar\varepsilon+\varepsilon_0 e^{-2\Gamma_{\max}t}.
\]
Hence the coherence-dominance condition~\eqref{eq:CD}
follows from
\[
\bar\varepsilon
\le
(A_\theta c_0-\varepsilon_0)e^{-2\Gamma_{\max}t},
\]
which is equivalent to \(t\le T^*\).

If \(\bar\varepsilon=0\), then
\[
e^{-kt}\varepsilon_0
\le
e^{-2\Gamma_{\max}t}\varepsilon_0
\le
A_\theta e^{-2\Gamma_{\max}t}c_0,
\]
using \(\varepsilon_0\le A_\theta c_0\).
Hence~\eqref{eq:CD} holds for all \(t\ge0\).
\end{proof}

\section{Consequences of the Coherence Window}\label{sec:corollaries}

\begin{corollary}[Low-temperature window]\label{cor:low-T}
Suppose the Davies generator satisfies quantum detailed balance and $\bar\varepsilon>0$ with Gibbs
weights $\pi_p\propto e^{-\beta E_p}$, so that $W_{ep}=W_{pe}e^{-\beta(E_e-E_p)}$.
Define $\Delta E:=\min\{E_e-E_p:p\in P,\,e\in Q,\,W_{ep}\ne 0\}>0$
and $\eta_\uparrow:=\max_{p\in P}\sum_{e\in Q}W_{pe}$.
Then $\mu\le e^{-\beta\Delta E}\eta_\uparrow$ and
\[
\bar\varepsilon\le\frac{\eta_\uparrow}{\eta}\,e^{-\beta\Delta E}.
\]
Assume $k\ge 2\Gamma_{\max}$, $A_\theta c_0>\varepsilon_0$, and
\begin{equation}
A_\theta c_0 - \varepsilon_0 \;>\; \frac{\eta_\uparrow}{\eta}\,e^{-\beta\Delta E},
\label{eq:low-T-pos}
\end{equation}
which in particular implies $A_\theta c_0-\varepsilon_0>\bar\varepsilon$
so that Proposition~\ref{prop:window} applies with $T^*>0$.
Under~\eqref{eq:low-T-pos}:
\[
T^*\ge\frac{1}{2\Gamma_{\max}}\!\left[\beta\Delta E
-\log\!\Bigl(\frac{\eta_\uparrow/\eta}{A_\theta c_0-\varepsilon_0}\Bigr)\right]>0.
\]
In particular, if $\eta_\uparrow/\eta$ is bounded independently of $\beta$,
condition~\eqref{eq:low-T-pos} holds for all sufficiently large $\beta$,
and $T^*$ grows at least linearly in $\beta$.
\end{corollary}
\begin{proof}
Detailed balance gives $\mu\le e^{-\beta\Delta E}\eta_\uparrow$,
so $\bar\varepsilon\le(\eta_\uparrow/\eta)e^{-\beta\Delta E}$.
Condition~\eqref{eq:low-T-pos} then gives $A_\theta c_0-\varepsilon_0>\bar\varepsilon$,
and Proposition~\ref{prop:window} applies, yielding $T^*>0$ with the stated lower bound.
\end{proof}

\begin{corollary}[Pure-state coherence ratio]\label{cor:pure}
Let $\rho_0=|\psi_0\rangle\langle\psi_0|$ with $|\psi_0\rangle=u+v$,
$u\in P\cH$, $v\in Q\cH$ (orthogonal).
Then $c_0=\|u\|^2\|v\|^2=(1-\varepsilon_0)\varepsilon_0$.
If $\varepsilon_0\le 1/2$, then $R(0)^2=(1-\varepsilon_0)/(2-\varepsilon_0)\ge 1/3$,
so condition~\eqref{eq:CD} holds at $t=0$ for any $\theta\le 1/\sqrt{3}$.
\end{corollary}
\begin{proof}
$B=P\rho_0 Q=|u\rangle\langle v|$, so $c_0=\|u\|^2\|v\|^2=(1-\varepsilon_0)\varepsilon_0$.
For $\varepsilon_0\le 1/2$: $c_0+\varepsilon_0=\varepsilon_0(2-\varepsilon_0)$ and
$R(0)^2=c_0/(c_0+\varepsilon_0)=(1-\varepsilon_0)/(2-\varepsilon_0)\ge 1/3$.
\end{proof}

\section{Coherence Decomposition and Multi-Sector Bounds}\label{sec:cps}

\subsection{Two-block decomposition}

Lemma~\ref{lem:BKM} has a structural interpretation as the coherence term in an
exact two-term decomposition of relative entropy.

\begin{theorem}[Coherence--population separation]\label{thm:CPS}
Let $\sigma$ satisfy $\Pi(\sigma)=\sigma$ (block-diagonal in $P\oplus Q$),
and let $\rho$ satisfy Assumption~\ref{ass:local} and $\varepsilon_Q(\rho)\le a_0/2$.

\begin{enumerate}
\item[\textup{(i)}] \textbf{Exact decomposition:}
\begin{equation}
D(\rho\|\sigma)
= \underbrace{D(\rho\|\Pi\rho)}_{D_{\mathrm{coh}}(\rho)}
+ \underbrace{D(\Pi\rho\|\sigma)}_{D_{\mathrm{pop}}(\rho,\sigma)}.
\label{eq:CPS-decomp}
\end{equation}
Here $D_{\mathrm{coh}}(\rho)=D(\rho\|\Pi\rho)$ equals the relative entropy of
coherence~\cite{BaumgratzCramerPlenio2014} in the $P\oplus Q$ basis.
\item[\textup{(ii)}] \textbf{Coherence lower bound:}
\begin{equation}
D_{\mathrm{coh}}(\rho)\;\ge\; c(\rho)\,\log\!\Bigl(\frac{a_0}{\varepsilon_Q(\rho)}\Bigr).
\label{eq:CPS-coh}
\end{equation}
\item[\textup{(iii)}] \textbf{Combined lower bound:}
\begin{equation}
D(\rho\|\sigma)\;\ge\;
c(\rho)\,\log\!\Bigl(\frac{a_0}{\varepsilon_Q(\rho)}\Bigr)+D(\Pi\rho\|\sigma).
\label{eq:CPS-main}
\end{equation}
\end{enumerate}
\end{theorem}

\begin{proof}
\textbf{Part~(i)} is the quantum Pythagorean theorem~\cite{Petz1986}
for the block-diagonal conditional expectation $\Pi$ and $\sigma$ in its fixed-point algebra:
$D(\rho\|\Pi\rho)+D(\Pi\rho\|\sigma)=\Tr[\rho\log\rho]-\Tr[\Pi\rho\log\sigma]=D(\rho\|\sigma)$,
using $\Tr[\rho X]=\Tr[\Pi\rho\cdot X]$ for block-diagonal $X$.
\textbf{Part~(ii)} is exactly~\eqref{eq:BKM-Picoh} from the proof of Lemma~\ref{lem:BKM}.
That equation bounds $D(\rho\|\Pi\rho)$ directly, without invoking the outer
reference $\sigma_\varepsilon$; it follows from Steps~2--3 of that proof alone.
\textbf{Part~(iii)} combines Part~(i), Part~(ii), and $D_{\rm pop}=D(\Pi\rho\|\sigma)\ge 0$.
\end{proof}

\begin{remark}
The decomposition~\eqref{eq:CPS-decomp} is standard~\cite{Petz1986},
while the lower bound~\eqref{eq:CPS-coh} gives an explicit support-sensitive
estimate for the coherence contribution.
\end{remark}

\subsection{Ordered multi-sector extension}\label{sec:multi-sector}

\begin{lemma}[Off-diagonal block under sequential pinching]\label{lem:block-preservation}
Let $P_1,\ldots,P_K$ be mutually orthogonal projections with $\sum_k P_k=\Id$.
Define $\rho^{(K)}:=\rho$ and $\rho^{(m-1)}:=\mathcal E_m(\rho^{(m)})$
where $\mathcal E_m(X):=P_mXP_m+(I-P_m)X(I-P_m)$.
Then for each $m=2,\ldots,K$:
\begin{equation}
(I-P_m)\,\rho^{(m)}\,P_m = \sum_{i=1}^{m-1}P_i\,\rho\,P_m.
\label{eq:block-id}
\end{equation}
Consequently, $C_m:=\|(I-P_m)\rho^{(m)}P_m\|_2^2=\sum_{i<m}\|P_i\rho P_m\|_2^2$.
\end{lemma}

\begin{proof}
The channel $\mathcal E_j$ preserves $P_i X P_{m'}$ whenever $j\notin\{i,m'\}$
(both $P_i$ and $P_{m'}$ lie in the range of $I-P_j$), and zeros out $P_i X P_{m'}$
if $j\in\{i,m'\}$.
After applying $\mathcal E_K,\ldots,\mathcal E_{m+1}$:
every block $P_j\rho P_m$ with $j>m$ is killed at step $\mathcal E_j$,
while every block $P_i\rho P_m$ with $i<m$ survives
(no $\mathcal E_\ell$ with $\ell>m$ has $\ell\in\{i,m\}$ for $i<m<\ell$).
The norm identity follows from $P_iP_j=0$ for $i\ne j$.
\end{proof}

\begin{theorem}[Ordered multi-sector bound]\label{thm:multi-block}
Under the notation of Lemma~\ref{lem:block-preservation}, let $\Pi_K(\sigma)=\sigma$.
For each $m=2,\ldots,K$, assume:
\begin{align}
a_m &:= \lambda_{\min}\!\bigl((I-P_m)\rho^{(m)}(I-P_m)\big|_{\mathrm{ran}(I-P_m)}\bigr) > 0,
\label{eq:am-cond}\\
\varepsilon_m &:= \Tr(P_m\rho P_m) \le a_m/2.
\label{eq:epsm-cond}
\end{align}
These conditions ensure that each sector satisfies a local support and small-population constraint compatible with the two-block coercivity estimate.
\begin{equation}
D(\rho\|\sigma)
\ge\sum_{m=2}^K C_m\,\log\!\Bigl(\frac{a_m}{\varepsilon_m}\Bigr)
+ D(\Pi_K\rho\|\sigma).
\label{eq:multi-block}
\end{equation}
\end{theorem}

\begin{proof}
\textbf{Telescoping.}
$\Pi_K\rho$ is fixed by each $\mathcal E_m$,
so the quantum Pythagorean theorem gives
$D(\rho^{(m)}\|\Pi_K\rho)=D(\rho^{(m)}\|\rho^{(m-1)})+D(\rho^{(m-1)}\|\Pi_K\rho)$.
Telescoping: $D(\rho\|\Pi_K\rho)=\sum_{m=2}^K D(\rho^{(m)}\|\rho^{(m-1)})$.
\textbf{Per-step bound.}
At step $m$, note that $\rho^{(m-1)}=\mathcal E_m(\rho^{(m)})=\Pi_{P_m}(\rho^{(m)})$
is the block-diagonal pinching of $\rho^{(m)}$ in the $(I-P_m)\oplus P_m$ split,
so $D(\rho^{(m)}\|\rho^{(m-1)})=D(\rho^{(m)}\|\Pi_{P_m}(\rho^{(m)}))$.
This is the relative entropy of coherence of $\rho^{(m)}$ in the $(I-P_m)\oplus P_m$ basis.
The off-diagonal block of $\rho^{(m)}$ in this split is $\sum_{i<m}P_i\rho P_m$
with squared Frobenius norm $C_m$ (Lemma~\ref{lem:block-preservation}).
The proof of Lemma~\ref{lem:BKM} (Steps 2--3, which do not involve the outer reference
$\sigma_\varepsilon$) applies verbatim with $a_0\leftarrow a_m$ and $\varepsilon_Q\leftarrow\varepsilon_m$:
data-processing to the SVD-adapted $2\times2$ blocks gives
$D(\rho^{(m)}\|\rho^{(m-1)})\ge C_m\log(a_m/\varepsilon_m)$
under conditions~\eqref{eq:am-cond}--\eqref{eq:epsm-cond}.
Adding $D(\Pi_K\rho\|\sigma)$ via the Pythagorean theorem yields~\eqref{eq:multi-block}.
\end{proof}

\begin{remark}[On the local conditions]\label{rem:local-conds}
Conditions~\eqref{eq:am-cond}--\eqref{eq:epsm-cond} require the support block
$(I-P_m)\rho^{(m)}(I-P_m)$ to be uniformly invertible and sector $m$ to carry less
than half the minimum weight of its complement.
These are genuine hypotheses; they are not automatic for $K\ge 3$.
In the two-sector case ($K=2$) they reduce to Assumption~\ref{ass:local} and $\varepsilon_Q\le a_0/2$.
For $K\ge 3$, a sufficient condition is a strict population hierarchy: if the sectors are
ordered so that $\Tr(P_k\rho P_k)$ decreases geometrically in $k$, the conditions
hold with constants determined by the geometric ratio.
\end{remark}

\subsection{Comparison with the Fawzi--Renner bound}\label{sec:FR-comparison}

We compare the coherence lower bound~\eqref{eq:CPS-coh} with the
Fawzi--Renner recoverability bound for the pinching channel.

\begin{lemma}[Petz recovery fixes the dephased state]\label{lem:petz-fixes}
Let $\sigma$ be a faithful (invertible) density matrix satisfying $\Pi(\sigma)=\sigma$.
The Petz recovery map for the pinching channel $\Pi$ with reference $\sigma$ is
\[
\mathcal R_{\sigma,\Pi}(X):=\sigma^{1/2}\Pi(\sigma^{-1/2}X\sigma^{-1/2})\sigma^{1/2},
\]
which is well-defined since $\sigma$ is invertible.
One has $\mathcal R_{\sigma,\Pi}(\Pi\rho)=\Pi\rho$,
and consequently $-2\log F(\rho,\mathcal R_{\sigma,\Pi}(\Pi\rho))=-2\log F(\rho,\Pi\rho)$.
\end{lemma}

\begin{proof}
Since $\sigma$ is faithful, $\sigma^{-1/2}$ is a bounded operator.
Both $\sigma^{-1/2}$ and $\Pi\rho$ are block-diagonal in the same $P\oplus Q$ decomposition,
so $Y:=\sigma^{-1/2}(\Pi\rho)\sigma^{-1/2}$ is block-diagonal, hence $\Pi(Y)=Y$, and
$\mathcal R_{\sigma,\Pi}(\Pi\rho)=\sigma^{1/2}Y\sigma^{1/2}=\Pi\rho$.
\end{proof}

\begin{theorem}[Near-boundary comparison with Fawzi--Renner]\label{thm:FR-comparison}
Let $\sigma$ be a faithful block-diagonal density matrix satisfying $\Pi(\sigma)=\sigma$,
and let $\rho$ satisfy:
\begin{enumerate}
\item[\textup{(i)}] $P\rho P\ge a_0 P$,
\item[\textup{(ii)}] $\varepsilon_Q(\rho)\le a_0\,e^{-4/a_0}$ (equivalently, $\log(a_0/\varepsilon_Q)\ge 4/a_0$).
\end{enumerate}
Then
\begin{equation}
c(\rho)\,\log\!\Bigl(\frac{a_0}{\varepsilon_Q(\rho)}\Bigr)
\;\ge\;
-2\log F\!\bigl(\rho,\,\mathcal R_{\sigma,\Pi}(\Pi\rho)\bigr).
\label{eq:FR-comparison}
\end{equation}
Moreover, for any family $\{\rho_s\}_{s>0}$ with $c(\rho_s)>0$ and $\varepsilon_Q(\rho_s)\to 0$
satisfying~\textup{(i)}-\textup{(ii)}, the ratio satisfies
\begin{equation}
\frac{c(\rho_s)\log(a_0/\varepsilon_Q(\rho_s))}{-2\log F(\rho_s,\mathcal R_{\sigma,\Pi}(\Pi\rho_s))}
\;\ge\;
\frac{a_0\,\log(a_0/\varepsilon_Q(\rho_s))}{4}
\;\xrightarrow{s\to 0}\;\infty.
\label{eq:ratio-diverges}
\end{equation}
\end{theorem}

\begin{proof}
By Lemma~\ref{lem:petz-fixes}, the right side of~\eqref{eq:FR-comparison}
equals $-2\log F(\rho,\Pi\rho)$.

\textbf{Upper bound on $-2\log F(\rho,\Pi\rho)$ via the SVD pinching.}
Recall the SVD pinching channel $\Phi$ from the proof of Lemma~\ref{lem:BKM}:
it maps $\rho$ to a direct sum over the subspaces $K_j=\operatorname{span}\{u_j,v_j\}$
and $K_\perp$.
These subspaces are mutually orthogonal: the $\{u_j\}$ are orthonormal in $P\cH$
and the $\{v_j\}$ are orthonormal in $Q\cH$ (both from the SVD of $B=P\rho Q$),
so $K_j\perp K_{j'}$ for $j\ne j'$ and $K_j\perp K_\perp$ by definition.
Fidelity is additive across orthogonal direct sums:
\[
F(\bigoplus_j X_j,\bigoplus_j Y_j)
=
\sum_j F(X_j,Y_j).
\]
By data processing for $\Phi$:
\begin{equation}
F(\rho,\Pi\rho)\ge F(\Phi(\rho),\Phi(\Pi\rho))
= \sum_j F(M_j,D_j) + F(M_\perp,D_\perp).
\label{eq:F-additive}
\end{equation}
Since $M_\perp=D_\perp$ (the perpendicular block is unchanged by $\Phi$),
$F(M_\perp,D_\perp)=\Tr(M_\perp)=1-\sum_j(a_j+c_j)$.

For each $2\times 2$ block, writing $t_j:=a_j+c_j$
and $M_j^{\rm n}:=M_j/t_j$ (normalized), the scaling property of fidelity gives
$F(M_j,D_j)=t_j F(M_j^{\rm n},D_j^{\rm n})$.
From the exact formula $F^2(M_j^{\rm n},D_j^{\rm n})
=1-2(a_j/t_j)(c_j/t_j)(1-\sqrt{1-(s_j/t_j)^2/(a_j/t_j)(c_j/t_j)})$
and the inequality $1-\sqrt{1-x}\le x$ ($x\in[0,1]$), one has
$1-F^2(M_j^{\rm n},D_j^{\rm n})\le 2(s_j/t_j)^2$, hence
$F(M_j^{\rm n},D_j^{\rm n})\ge\sqrt{1-2(s_j/t_j)^2}\ge 1-(s_j/t_j)^2$,
so $F(M_j,D_j)\ge t_j-s_j^2/t_j$.

Substituting into~\eqref{eq:F-additive}:
\begin{equation}
F(\rho,\Pi\rho)
\ge \sum_j\!(t_j - s_j^2/t_j) + (1-\sum_j t_j)
= 1 - \sum_j s_j^2/t_j.
\label{eq:F-lower}
\end{equation}
Using $t_j=a_j+c_j\ge a_j\ge a_0$ (from condition~(i)):
$\sum_j s_j^2/t_j\le\sum_j s_j^2/a_0=c(\rho)/a_0$.
Thus $F(\rho,\Pi\rho)\ge 1-c(\rho)/a_0$.

Since $c(\rho)\le\varepsilon_Q(\rho)\le a_0 e^{-4/a_0}$ (the first inequality
holds because $c(\rho)=\sum_j s_j^2$ and $s_j^2\le a_jc_j$, hence
$\sum_j s_j^2\le\sum_j c_j\le\Tr(C)=\varepsilon_Q$, using $a_j\le 1$),
we have $c(\rho)/a_0\le e^{-4/a_0}\le 1/2$.
Hence, using $-\log(1-u)\le u/(1-u)\le 2u$ for $u\le 1/2$:
\begin{equation}
-2\log F(\rho,\Pi\rho)\le 4c(\rho)/a_0.
\label{eq:FR-upper}
\end{equation}

\textbf{Our lower bound.}
Condition~(ii) gives $\log(a_0/\varepsilon_Q)\ge 4/a_0$, so
$c(\rho)\log(a_0/\varepsilon_Q)\ge 4c(\rho)/a_0\ge -2\log F(\rho,\Pi\rho)$,
establishing~\eqref{eq:FR-comparison}.
For the ratio~\eqref{eq:ratio-diverges}: when $c(\rho_s)>0$, we have
$-2\log F(\rho_s,\Pi\rho_s)\le 4c(\rho_s)/a_0$ by~\eqref{eq:FR-upper},
so $-2\log F>0$ and the ratio satisfies
\[
\frac{c(\rho_s)\log(a_0/\varepsilon_Q(\rho_s))}{-2\log F(\rho_s,\Pi\rho_s)}
\ge \frac{c(\rho_s)\log(a_0/\varepsilon_Q(\rho_s))}{4c(\rho_s)/a_0}
= \frac{a_0\log(a_0/\varepsilon_Q(\rho_s))}{4}\to\infty.\]
\end{proof}

\begin{remark}[Scope of Theorem~\ref{thm:FR-comparison}]\label{rem:FR-scope}
The comparison requires $\varepsilon_Q\le a_0 e^{-4/a_0}$, a genuine near-boundary
restriction (for $a_0=0.8$, this is approximately $\varepsilon_Q\le 5\times 10^{-3}$).
Under the dynamics of Proposition~\ref{prop:Davies}, this is satisfied once the
kernel population has decayed sufficiently, roughly after time
$t\ge k^{-1}\log(\varepsilon_0/(a_0 e^{-4/a_0}))$.
The ratio bound~\eqref{eq:ratio-diverges} shows that the logarithmic enhancement
grows unboundedly as \(\varepsilon_Q\to0\):
our explicit bound captures a $\log(a_0/\varepsilon_Q)\to\infty$ factor that
the Petz-recovery remainder does not.
\end{remark}

\section{Relation to the Literature}\label{sec:literature}

\paragraph{BKM metrics.}
The BKM quadratic form and its spectral representation are studied
in~\cite{Petz1996,Lesniewski1999}.
The approach here differs from standard BKM coercivity estimates in that the
argument does not bound the full form by the smallest eigenvalue of the
interpolating state. Instead, the pinching reduction isolates effective
\(2\times2\) blocks controlled by the coherence--population contrast.

\paragraph{MLSI and entropy methods.}
Modified logarithmic Sobolev inequalities for Davies generators are established
in~\cite{KastoryanoTemme2013,BardetCapelGaoLuciaPerezGarciRouze2021,WirthZhang2021}.
The present work is complementary: given a finite-time entropy bound of the form
in Assumption~\ref{ass:MLSI}, we derive corresponding entropy--activation bounds
near rank-deficient stationary states.

\paragraph{Support-sensitive coercivity and coherence entropy.}
Logarithmic entropy divergences near vanishing eigenvalues are classical.
The bounds obtained here are support-sensitive: the entropy cost of
cross-boundary coherence is amplified by the small population scale
\(\varepsilon_Q(\rho)\).
Standard estimates for the relative entropy of coherence
\[
C_{\rm rel}(\rho):=D(\rho\|\Pi\rho)
\]
relate it to norm-based coherence measures such as trace-distance or
\(l_1\)-coherence bounds
(see~\cite{BaumgratzCramerPlenio2014,ZhuHayashiChen2018,WinterYang2016,StreltsovAdessoPlenio2017}).
In contrast, Lemma~\ref{lem:BKM} yields the support-sensitive estimate
\[
C_{\rm rel}(\rho)
\ge
c(\rho)\log\!\Bigl(\frac{a_0}{\varepsilon_Q(\rho)}\Bigr),
\]
which becomes logarithmically enhanced near rank-deficient boundaries.
For Davies semigroups satisfying
Assumptions~\ref{ass:secular}--\ref{ass:rates},
the resulting activation bounds combine this coercivity estimate with
secular decoupling and population-rate control.

\paragraph{Recoverability.}
The Fawzi--Renner inequality~\cite{FawziRenner2015}
gives a universal lower bound on the relative entropy deficit under a channel
in terms of a recovery fidelity.
Theorem~\ref{thm:FR-comparison} complements this with an explicit lower bound
displaying logarithmic enhancement in the near-boundary regime
\(
\varepsilon_Q(\rho)\to0
\).

\section{Discussion}\label{sec:discussion}

\paragraph{Summary.}
Section~\ref{sec:entropy-activation} derives entropy--activation bounds from
a support-sensitive coherence modulus near rank-deficient boundaries.
Theorem~\ref{thm:split} separates coherence and population contributions
without requiring coherence dominance, while
Theorem~\ref{thm:absorbing} gives an explicit absorbing-regime estimate.
For Davies semigroups satisfying the additional assumptions of
Section~\ref{sec:main}, Corollary~\ref{thm:main}
applies these structural bounds along trajectories for which the local support,
coherence-dominance, and entropy-budget hypotheses remain valid.
The coherence--population decomposition of
Theorem~\ref{thm:CPS} applies more generally to arbitrary block-diagonal
reference states.

\paragraph{Limitations.}
The results do not apply to incoherent initial states (\(c_0=0\)),
to times beyond the coherence-dominance window \(T^*\),
or to models where Assumption~\ref{ass:secular} fails,
for instance due to nearly degenerate cross-boundary Bohr frequencies.
The local condition~\eqref{eq:LC} and the coherence-dominance condition~\eqref{eq:CD}
are genuine trajectory-level hypotheses.
The logarithmic enhancement mechanism is specific to the near-boundary regime
and does not yield a universal comparison with Pinsker-type inequalities.

\section{Conclusion}\label{sec:conclusion}

The present work studies entropy--activation bounds near rank-deficient
boundaries for finite-dimensional Davies-type dynamics.
The central mechanism is a support-sensitive BKM coercivity estimate obtained
by reducing, through pinching and data processing, to effective \(2\times2\)
BKM blocks controlled by the coherence--population contrast.

Under the stated local and dynamical hypotheses, this coercivity estimate
yields logarithmically enhanced entropy--activation bounds for trajectories
with nontrivial cross-boundary coherence.
The coherence--population decomposition theorem further identifies the
coercive term as the coherence contribution in an exact decomposition of
relative entropy, and leads to a comparison with the Fawzi--Renner
recoverability bound in the near-boundary regime.

\end{document}